%% file: main.tex
\documentclass[11pt]{article}

\usepackage[letterpaper, margin=1in]{geometry}
\usepackage[T1]{fontenc}
\usepackage{natbib}
\usepackage[colorlinks=true, allcolors=blue]{hyperref}

\input{packages}

\usepackage[]{shortcuts}

\title{Polylogarithmic Collective Tree Exploration}
\author{Romain Cosson \\ New York University \and Laurent Massoulié \\ Inria, Paris}
\date{}

\begin{document}
\maketitle
\begin{abstract}
We study asynchronous collective tree exploration, where $k$ agents with unrestricted communication start at the root of an unknown tree and discover edges online. At each step, an adversary chooses which agent moves. We give a deterministic algorithm that explores any tree with $n$ nodes and depth $D$ in at most
\[
2n+O\!\left(k\log^2(k)D\right)
\]
moves, matching known lower bounds up to a constant factor. As a direct consequence, we obtain a near-optimal competitive ratio of $O(\log^2 k)$ for synchronous collective tree exploration, where all agents move at each round. The proof relies on a multiscale power regularizer that may be of independent interest.
\end{abstract}


\section{Introduction}\label{sec:introduction}
Consider a team of $k$ agents initially located at the root of an unknown tree. The agents can communicate without restriction, but they discover edges online: specifically, an edge is revealed to the team when an agent becomes adjacent to it. How should they coordinate to visit all nodes as fast as possible? In other words, \textit{how to parallelize tree exploration}?

In \textit{asynchronous collective tree exploration} (ACTE), introduced in \cite{cosson2024breaking}, the agents' speeds are controlled by an adversary. Specifically, at each step, the adversary selects one agent, and the algorithm chooses an adjacent edge for that agent to traverse. The activation sequence is arbitrary, and an algorithm's performance is measured by the maximum number of moves needed to explore the tree. The goal is to obtain a guarantee of the form 
\[
2n + O(r(k)\cdot kD)\quad\text{moves,}
\]
on all trees with $n$ nodes and depth $D$, where $r(k)$ is called the \textit{regret}.\footnote{The term $2n$ in the definition is justified as follows. To visit all nodes, any offline algorithm (i.e., one that sees the tree and the activation sequence) would require at least $2n - kD-2$ moves. This is because, after exploring the tree, all agents could return to the root in $kD$ moves, and then all $n-1$ edges would have been traversed at least twice. The term in $O(r(k)\cdot kD)$ is therefore an additive overhead due to the online nature of the feedback. There, multiplying $r(k)$ by $k$ is a natural renormalization in light of the synchronous model, as explained in the next paragraph.}

\paragraph{Main result.}
In this paper, we resolve the problem of asynchronous collective tree exploration under unrestricted communication up to a constant factor. 

\begin{theorem}[Asynchronous collective tree exploration]\label{thm:main-acte}
There exists a deterministic asynchronous exploration algorithm, with unrestricted communication, that explores any rooted tree with $n$ nodes and depth $D$ using $k\ge2$ agents in at most
\[
2n+O\!\left(k\log^2(k)D\right)
\]
moves, for every activation sequence. The regret is $O(\log^2(k))$, which is order optimal.
\end{theorem}

The classical setting of \textit{synchronous} collective tree exploration, introduced in \cite{fraigniaud2006collective}, allows all agents to move once per round. Our asynchronous guarantee yields a synchronous one via a classic reduction: the synchronous model corresponds to the adversary using a cyclic activation order, which provides $k$ moves per round \cite{cosson2024collective}. 

\begin{theorem}[Synchronous collective tree exploration]\label{cor:main-cte}
There exists a deterministic synchronous collective tree exploration algorithm, with unrestricted communication, that explores any rooted tree with $n$ nodes and depth $D$ using $k\geq 2$ agents in at most
\[
\frac{2n}{k}+O\!\left(\log^2(k)D\right) 
\]
rounds. The algorithm is $O(\log^2(k))$-competitive.
\end{theorem}

\paragraph{Technical challenges.} A useful point of comparison is the recent progress on the analysis of randomized algorithms for online metric problems such as $k$-server \cite{bubeck2018k}, metrical task systems \cite{bubeck2021metrical}, layered graph traversal \cite{bubeck2022shortest}, and ultrametric traversal \cite{bai2025unweighted}. In these works, \textit{multiscale entropic regularization} and the \textit{mirror-descent framework} played a central role in obtaining or improving competitive guarantees. Each of these problems required the application of multiscale entropic regularization in different ways, and despite several breakthroughs, the approach still lacks a unified understanding.

Collective tree exploration differs from all these problems in that it is typically studied for deterministic algorithms. Yet, conceptually, the swarm of $k$ agents can be interpreted as a distribution: specifically, one where each agent corresponds to a probability mass of $1/k$. This discrete structure is not suited to a direct application of continuous methods such as the mirror descent framework and multiscale entropic regularization \cite{bubeck2018k}. Therefore, we use the tree-mining game of \cite{cosson2024collective} to which asynchronous collective tree exploration reduces. The tree-mining strategy we present involves a new deterministic repair rule that rounds a fractional solution to a convex optimization problem.


A surprising aspect of the algorithm is that it relies on a multiscale power regularizer with exponent $1+\varepsilon$ for $\varepsilon=\Theta(1/\log k)$, rather than an entropic regularizer. On ultrametrics (i.e., when all leaves are at the same distance from the root) and in the limit $\epsilon \rightarrow 0$, the minimizer of the multiscale power regularizer coincides with the minimizer of the multiscale entropic regularizer due to the identity
\[
\frac{x^{\epsilon} - 1}{\epsilon} \rightarrow_{\epsilon \rightarrow 0} \ln(x), \qquad \forall x>0,
\]
but that is not the case on general trees. The power regularizer with small $\epsilon>0$ relates directly to Tsallis entropy in physics and learning theory, and we believe it could have other applications in online metric problems.

\paragraph{Related works.} The problem of (synchronous) collective tree exploration was introduced by \cite{fraigniaud2006collective}, who set as the main goal the design of algorithms exploring any tree with $n$ nodes and depth $D$ in at most $ c(k)\left(n/k + D\right)$ synchronous rounds, where $c(k)$ is called the competitive ratio.
This initial measure of performance was motivated by the fact that the number of time steps required by any team of $k$ agents to go through all edges of a tree $T$ satisfies $\text{OPT}(T)  = \Theta\left(\frac{n}{k}+D\right)$ (see \cite{fraigniaud2006collective} for details). Thus, $c(k)$ can be interpreted as a multiplicative factor the team pays for not knowing the tree, i.e., a \textit{price of exploration}. 

The original work of \cite{fraigniaud2006collective} proposed an $\Ocal(k/\log k)$-competitive algorithm that uses only limited communication between the agents, and conjectured that a constant competitive ratio was possible with complete communication. This conjecture was quickly disproved by \cite{dynia2007robots} with the lower bound $c(k) = \Omega(\log k /\log \log k)$. 

The competitive ratio of collective tree exploration was improved to $\Ocal(k/\exp(\sqrt{\log 2\log k}))$ by \cite{cosson2024breaking}, using a recursive algorithm, and then to $\Ocal(\sqrt k)$ by \cite{cosson2024collective}, using a quadratic regularizer. These works introduced the (algorithmically harder) setting of asynchronous collective tree exploration (ACTE) in which the measure of performance arising most naturally is the \textit{regret}. 
In this work, we obtain a regret guarantee matching the $\Omega(\log^2 k)$ lower bound established in \cite{cosson2025asynchronous}. Asynchronous regret guarantees translate to the weaker synchronous setting, where they yield the best competitive guarantees to date. For example, \cite{cosson2024collective} produced an asynchronous guarantee of the form $2n + O(k^2D)$, i.e., with regret $r(k) =O(k)$, inducing a synchronous guarantee of $2n/k + O(kD)$, which in turn leads to a $O(\sqrt{k})$-competitive algorithm. The present paper provides $O(\log^2 k)$-competitive algorithm, marking an exponential improvement for synchronous exploration as well.

Collective tree exploration has also been studied under additional assumptions. 
For trees that can be embedded in the $2$-dimensional grid, \cite{dynia2006smart} obtained an algorithm finishing in $\Ocal(\sqrt{D}(\frac{n}{k}+D))$ rounds. 
In the regime where the number of agents $k$ scales with the input tree size, $k\geq Dn^c$ for some constant $c>1$, \cite{dereniowski2015fast} proposed an exploration algorithm finishing in $\frac{c}{c-1}D+o(D)$ rounds. 
When the regret is allowed to depend superlinearly on $D$ \cite{brass2011multirobot, cosson2023efficient} also proposed other algorithms. 
Surprisingly, the problem of collective exploration on arbitrary graphs remains much less understood, and it is not even known whether an $o(k)$-competitive algorithm exists \cite{dereniowski2015fast,brass2014improved}.\footnote{Using a single agent to perform a DFS is $\Theta(k)$ competitive.}

The broader question of exploring an unknown environment with multiple agents is a recurring theme in computer science and robotics \cite{browne2012survey,parker2016multiple,das2019graph}. This challenge also naturally arises in modern AI applications \cite{koh2024tree,yamada2025ai,xin2025deepseek,hubert2026olympiad,hadfield2025multiagent}.\footnote{Collective tree exploration does not capture all such collective search problems. In particular, if an agent (or computing unit) can jump between distant states in constant time, as in the \textit{RAM model}, then parallel tree exploration becomes essentially trivial. Instead, CTE captures the \textit{mobile model}, in which distance between nodes represents a movement cost. This can be motivated in typical LLM use cases when accounting for the fact that moving between ``nearby'' contexts in the prefix tree is cheaper because it requires fewer KV-cache updates \cite{zheng2024sglang}.}
For example, the `Tree of Thoughts' paradigm \cite{yao2023tree, wei2022chain} has shown successes in complex reasoning and problem-solving with large language models. In that framework, the tree represents different reasoning paths (i.e., potential proofs), some of which lead to the desired conclusion. Similarly, \cite{chen2023walking} proposes navigating a large context via a combinatorial structure, such as a tree. 
More generally, we believe that the domain of online algorithms, that was initially motivated by the insights it provides on subroutines of operating systems and data structures, can serve as a theoretical framework for the interaction of autonomous intelligent agents.


\section{Notations and preliminaries}\label{sec:notations}
\subsection{Notations}
All logarithms are natural. Trees are finite and rooted, with root $r$. For a node $u\ne r$, we denote its parent by $p_u$ and the set of its children by $C_u$. We write $u\preceq v$ when $u$ is a descendant of $v$, allowing equality, and $u\prec v$ for strict descent. For $u\preceq v$, the path $u\rightarrow v$ includes $u$ and excludes $v$. The set of leaves is $\Lcal(T)$, abbreviated $\Lcal$, and $\Lcal_u=\{\ell\in\Lcal:\ell\preceq u\}$. In the rest of the paper, trees are weighted. The length of edge $(p_u,u)$ is $d_u$, and $d(u,v)$ is the length of the unique path between $u$ and $v$. The depth of a weighted tree is 
\[
D=\max_{\ell\in \Lcal}d(r,\ell).
\]

\paragraph{Configurations.}
A fractional leaf configuration is a vector $\by=(y_\ell)_{\ell\in\Lcal}$ with nonnegative entries and total mass $\sum_{\ell\in\Lcal}y_\ell=k$. Its coordinates are extended to all nodes by
\[
y_u=\sum_{\ell\in\Lcal_u}y_\ell.
\]
An integral configuration $\bx$ must have discrete masses $x_\ell \in \{1,\dots, k\}$ and a total mass $\sum_{\ell\in\Lcal}x_\ell=k$. 

\paragraph{Transportation and movement.}
For any two configurations $\by,\by'$ on a fixed weighted tree, their \textit{optimal transportation distance} is
\[
\OT_d(\by,\by')=\sum_{u\ne r}d_u|y_u-y_u'|.
\]
When a configuration indexed by time $\by(t)$ is differentiable, its \textit{instantaneous movement cost} is defined by 
\[
\|\dot{\by}(t)\|_d=\sum_{u\ne r}d_u|\dot y_u(t)|,
\]
where the derivative $\dot \by(t)$ can be interpreted as a flow between the leaves of the tree (note that $\dot \by$ satisfies Kirchhoff's current conservation law on all non-leaf nodes, including the root). The cumulative cost $C_{\by}(t)$ is the integral (or sum, when applicable) of all movement costs of an evolving configuration $\by(\cdot)$ up to time $t$. We omit time arguments when clear from context.

\subsection{The continuous tree-mining game}\label{sec:ctm-background}
We now present the continuous tree-mining game (CTM) introduced in \cite{cosson2024collective}. The game state consists of a weighted tree $T(t)$, controlled by the adversary, and an integral leaf configuration $\bx(t)$ of $k$ miners, controlled by the player. Starting with all miners at the root, the tree may evolve in the following ways:
\begin{itemize}
\item \textbf{Elongation.} The adversary may grow the edge length $d_l$ at a leaf $l$ with $x_l\ge2$ at unit rate. The player pays $x_l$ per unit of growth, in addition to any redistribution movement cost.
\item \textbf{Fork.} The adversary may fork a leaf $l$ with $x_l\ge3$ creating $m\in\{2,\ldots,x_l-1\}$ new children. The player chooses their common edge length $\delta\in(0,1]$ and moves miners at $l$ to its children.
\item \textbf{Deletion.} The adversary may delete a leaf $l$. The player must move its miners to surviving leaves. 
\end{itemize}
The player may update the miner's configuration $\bx(t)$ at any time, and pays the cumulative movement cost associated with redistributing miners between leaves. Over the course of the game, every nonroot internal vertex has at least two children. A vertex with degree two after deletion of a child is suppressed by replacing its two incident edges with one edge of their combined length.

\begin{theorem}[Tree-mining reduction, \cite{cosson2024collective}]\label{thm:ctm-reduction}
If a continuous tree-mining strategy can maintain a cumulative cost of at most $f(k,D)$ while the underlying tree's depth is bounded by $D$, then there is an asynchronous collective tree algorithm exploring any tree with $n$ nodes and depth $D$ in at most 
\[
2n+f(k,D)\qquad \text{moves.}
\]
\end{theorem}

\begin{theorem}[Continuous tree-mining bound]\label{thm:ctm-bound}
There exists a tree-mining strategy whose cumulative cost is bounded by 
\[
f(k,D)
\le\frac{8192}{\log(4/3)}\,k\log^2(2k)\,D = O(k\log^2(k)D).
\]
\end{theorem}

The proof of Theorem~\ref{thm:main-acte}
is obtained from Theorem~\ref{thm:ctm-bound} by applying Theorem~\ref{thm:ctm-reduction} with a tree-mining strategy defined in the next section, where Propositions~\ref{lem:rounding} and~\ref{lem:potential}, bound its cumulative cost by $O(k\log^2(k)D)$. 

\section{Algorithm and analysis}

\subsection{Algorithm description}
Fix $k\ge2$, and define
\[
\varepsilon=\frac{\log(4/3)}{\log(2k)},\qquad p=1+\varepsilon.
\]
Thus
\[
(2k)^\varepsilon=\frac43,\qquad
k^p\le2k,\qquad
0<\varepsilon<1,\qquad
\varepsilon^{-1}=\Theta(\log k).
\]

For a weighted tree $T$, define the \textit{multiscale power regularizer},
\begin{equation}
    \Phi_T(\by)=\sum_{u\ne r}d_u y_u^p.
\end{equation}
For the current tree $T(t)$, let $\by(t)$ be the fractional minimizer, 
\begin{equation}
    \by(t) \in \arg\min_{\substack{\by\ge0\\\sum_\ell y_\ell=k}}\Phi_T(\by).\label{eq:optim}
\end{equation}
\begin{lemma}
    Equation \eqref{eq:optim} uniquely defines $\by(t)$, which has strictly positive mass at every leaf.  Furthermore, during the continuous elongation of a leaf $l$, $\by(t)$ is differentiable, with 
    \[
    \forall u\in l \rightarrow r : \dot y_u\leq 0 \qquad \text{and}  \qquad \forall u\not \in l \rightarrow r :\dot y_u \geq 0.
    \] 
\end{lemma}
\begin{proof}
Since all edge lengths are positive, the minimizer $\by(t)$ is unique by strict convexity of $\Phi_T(\cdot)$. The mass is positive at every leaf because moving mass to a leaf with $y_\ell=0$ from the nearest populated leaf $\ell'\neq \ell$ decreases the objective at rate $p\sum_{u\in\ell'\rightarrow v}d_uy_u^{\epsilon}>0$, where $v$ is the lowest common ancestor of $\ell$ and $\ell'$.  During the continuous elongation of a leaf $l$, $\by(t)$ is differentiable, with $\dot y_l\leq 0$ and $\forall \ell\neq l :\dot y_\ell \geq 0$  \cite[Proposition 2.6]{cosson2024collective}.  Thus for $u\not \in l\rightarrow r$, $\dot y_u=\sum_{\substack{\ell\in\Lcal_u}}\dot y_\ell\ge 0$ and the sign of $\dot y_u$ for $u\in l\rightarrow r$ follows from Kirchhoff's current conservation law at non-leaf nodes. 
\end{proof}

\paragraph{Player's strategy.} We now describe the CTM strategy $\bx(t)$, which is initialized with all $k$ miners located at the root and is updated upon one of the following three events.

\textbf{a) Repair.}
Whenever a leaf $\ell$ satisfies $x_\ell-y_\ell\ge3/2$, denote by $v$ the first ancestor of $\ell$ satisfying $x_v-y_v<3/2$.
From $v$, descend through the underloaded children satisfying
$x_u-y_u<0$ until reaching a leaf $\ell'$, and move one miner
from $\ell$ to $\ell'$.
Repeat until every leaf satisfies $x_\ell-y_\ell<3/2$.

\textbf{b) Fork.}
When a leaf $\ell$ receives $m$ children, split its miners as evenly
as possible among them, with groups of $\{\left\lceil\frac{x_\ell}{m}\right\rceil, \left\lfloor\frac{x_\ell}{m}\right\rfloor\}$. The fork length
$\delta\in(0,1]$ is chosen to be sufficiently small (see Appendix~\ref{sec:rounding}). 

\textbf{c) Deletion.}
At the deletion of $l$, send its miners to any surviving leaf $\ell$ below
a sibling of $l$,
remove the deleted edge,
then perform repairs at $\ell$.

The CTM algorithm satisfies the following property, shown in Appendix \ref{sec:rounding}, and which implies that we can concentrate our effort on the cumulative cost of $\by(t)$.

\begin{restatable}{proposition}{roundinglemma}\label{lem:rounding}
The CTM algorithm $\bx(t)$ is well-defined and enforces at all times the invariant
\begin{equation}
    y_\ell\geq \frac12,\qquad
x_\ell\ge1,\qquad
x_\ell< y_\ell + \frac32.\tag{I}
\label{eq:invariants}
\end{equation}
Furthermore, denoting by
$C_{\bx}(t)$ (resp. $C_{\by}(t)$) the cumulative movement cost
of configuration $\bx(\cdot)$ (resp. $\by(\cdot)$) at time $t$,
we have
\begin{equation}
    C_{\bx}(t)\le128C_{\by}(t).\label{eq:ybx}
\end{equation}
\end{restatable}

\subsection{Proof strategy for bounding $C_{\by}(t)$}
The main technical lemma of the paper is the following inequality, which is obvious at time $t=0$, and remains true through elongations (Section \ref{sec:elongation}), deletions (Section \ref{sec:deletion}), and forks (Section \ref{sec:forks}). 
\begin{proposition}\label{lem:potential}
    Writing $\Phi(t)=\Phi_{T(t)}(\by(t))$, the following inequality remains true at all times,
\begin{equation}\label{eq:goal}
        C_y(t) \leq \frac{32\log(2k)}{\varepsilon}
\Phi(t).
\end{equation}
\end{proposition}

Note that, by considering the configuration placing all mass on any one leaf of the tree, we have,
\begin{equation}
    \Phi(t) \le k^pD\le2kD.\label{eq:phi}
\end{equation}
Therefore, Proposition \ref{lem:rounding} and Proposition \ref{lem:potential} imply that
\[
\begin{aligned}
C_{\bx}(t)
&\le128C_{\by}(t)\\
&\le\frac{4096\log(2k)}{\varepsilon}\Phi(t)=O\!\left(k\log^2(2k)D\right).\\
\end{aligned}
\]
where we used $\Phi(t)\le2kD$ and
$\varepsilon=\log(4/3)/\log(2k)$.

We start the proof Proposition \ref{lem:potential} with two general-purpose lemmas. 

\begin{lemma}[Geometry of the fractional minimizer]\label{lem:geometry}
The first-order condition in \eqref{eq:optim} rewrites as
\[
\forall \ell \in \Lcal: \sum_{u\in\ell\rightarrow r}d_u y_u^\varepsilon
 =\frac{\Phi(t)}{k}.
\tag{FO}\label{eq:FO}
\]
Define for any $u\neq r$ the approximate edge length
\[
\rho_u=d_u y_u^\varepsilon,\quad \text{which satisfies} \quad \frac12d_u\le\rho_u\le2d_u \quad \text{under \eqref{eq:invariants}},
\]
then, the approximate height of $u$ given by
\[
\qquad
h_u=\sum_{v\in\ell\rightarrow u}\rho_v, \quad \text{where $\ell\preceq u$ is any descendant leaf,}
\]
does not depend on the choice of $\ell$; in particular, $h_r=\Phi(t)/k$ by \eqref{eq:FO}.
\end{lemma}

\begin{proof}
The leaf gradients of $\Phi$ are all equal at the minimizer of \eqref{eq:optim}, 
\[
\forall \ell \in\Lcal : p\sum_{u\in \ell \rightarrow r}d_u y_u^\varepsilon = \lambda
\]
for some $\lambda\in \mathbb{R}$. Furthermore, we have
\[
k\lambda = \left(\sum_{\ell \in \Lcal}y_\ell\right) \lambda = p\sum_{\ell\in\Lcal}y_\ell\sum_{u\in \ell \rightarrow r}d_u y_u^\epsilon = p\Phi(t),
\]
which proves \eqref{eq:FO}. Under \eqref{eq:invariants}, every subtree mass lies in $[1/2,k]$. Hence
\[
\frac12\le 2^{-\varepsilon}\le y_u^\varepsilon
\le k^\varepsilon\le2,
\]
which proves the comparison between $\rho_u$ and $d_u$. 
\end{proof}

\begin{lemma}
    Under \eqref{eq:invariants}, any two distinct leaves $\ell,\ell'\preceq v$ satisfy
\begin{equation}
    \frac{d(v,\ell')}{d(v,\ell)}\le\frac43.
\label{eq:geometry}
\end{equation}
\end{lemma}
\begin{proof}
    Both paths from $v$ to $\ell$ and $\ell'$ have approximate length $h_v$, using $\frac{1}{2}\leq y_u\leq k$ under \eqref{eq:invariants}, we have, 
    \[
    d(v,\ell) = \sum_{u\in \ell\rightarrow v}d_u \leq 2^\epsilon h_v\leq (2k)^\epsilon \sum_{u\in \ell'\rightarrow v}d_u  = \frac{4}{3}d(\ell',v).
    \]
\end{proof}

\subsection{Leaf elongations: an electrical argument}\label{sec:elongation}
In this section, we show that \eqref{eq:goal} is preserved during elongations of an edge $(p_l,l)$, i.e., that
\[
\boxed{\dot C_y(t):= y_l+ ||\dot \by(t)||_d \leq \frac{32\log(2k)}{\varepsilon}\dot \Phi(t).} \tag{Elongation Eq.}\label{eq:elong}
\]
By the chain rule, when $(p_l,l)$ is extended at unit rate $\dot d_l =1$, we have, 
\[
\dot\Phi(t) = y_l^p +  \nabla_\by\Phi(t) \cdot \dot  \by(t) = y_l^p
\]
since $\nabla_\by\Phi(t) \cdot \dot  \by(t) = 0$ by \eqref{eq:FO}. We therefore turn to expressing the left-hand side of \eqref{eq:elong} in terms of $y_l^p$. We start with an electrical perspective given in Lemma \ref{lem:electrical} and Lemma \ref{lem:voltage} below.

\begin{lemma}[Electrical representation]\label{lem:electrical}
In the elongation of leaf $l$ at unit rate, for every leaf $\ell$, we have,
\[
\sum_{u\in\ell\rightarrow r}\frac{\rho_u}{y_u}\dot y_u
=\frac{\dot \Phi(t)/k-y_l^\varepsilon\mathbf1_{\{\ell=l\}}}{\varepsilon}.
\]
Therefore, $\dot \by$ can be interpreted as the current obtained in the resistive network on $T$ where each edge $(p_u,u)$ for $u\neq r$ has resistance
\begin{equation}
    r_u=\frac{\rho_u}{y_u} \qquad \text{and leaf voltages are}\qquad \forall \ell \neq l : V_\ell = 0 \quad\text{and}  \quad V_l = y_l^\varepsilon/{\varepsilon}.\label{eq:network}
\end{equation}
In this network, the effective resistance between $u$ and $\Lcal_u$ is equal to 
\[
R_u = \frac{h_u}{y_u}.
\]
\end{lemma}

\begin{proof}
Differentiating \eqref{eq:FO} for every leaf $\ell$, gives
\[
\sum_{u\in\ell\rightarrow r}\frac{\rho_u}{y_u}\dot y_u
=\frac{\dot \Phi(t)/k-y_l^\varepsilon\mathbf1_{\{\ell=l\}}}{\varepsilon},
\]
which implies $\dot \by(t)$ is the unique solution of Ohm's and 
Kirchhoff's laws for the specified resistances and voltages (note that voltages are defined up to a constant shift).

For the computation of effective resistances $R_u$, we argue inductively from the leaves, which have zero effective resistance. Considering a node $u$ that is not a leaf, the induction hypothesis allows us to see the branch $(u,c)$ for $c\in C_u$ as having effective resistance $R_c + r_c = \frac{h_u}{y_c}$. Then, for all the children $c\in C_u$, the parallel conductances add up to $\sum_{c\in C_u}y_c/h_u=y_u/h_u$, proving that the conductance of the subtree rooted at $u$ is $y_u/h_u$; thus, $R_u = h_u/y_u$.
\end{proof}

\begin{lemma}[Voltage identity]\label{lem:voltage}
    In the electrical network described in \eqref{eq:network}, the voltage at node $v\succeq l$ satisfies the following identity,
    \begin{equation}
    y_v V_v - h_v\dot y_v = y_l V_l, \label{eq:open-circuit}
\end{equation}
and since $\dot y_v\leq 0$, we have $V_v \leq \frac{y_l}{y_v}V_l$. 
\end{lemma}
\begin{proof}
    For a node $v\succeq l$, we obtain \eqref{eq:open-circuit}
by induction, starting from $l$ on which \eqref{eq:open-circuit} is obvious. Suppose that \eqref{eq:open-circuit} is true for $u$, the child of $v$ that is an ancestor of $l$. Using Ohm's law on the branches associated to $c\in C_v\setminus\{u\}$, we have that
\[
\dot y_v = \sum_{c\in C_v}\dot y_c = \dot y_{u} + V_v\sum_{c\in C_v \setminus\{u\}}\frac{y_c}{h_v} = \dot y_{u} + V_v \frac{y_v-y_{u}}{h_v}
\]
and using Ohm's law on edge $(u,v)$ we have that 
\[
V_v - V_u = r_u \dot y_u = \frac{\rho_u}{y_u}\dot y_u
\]
Using those two identities, as well as the induction hypothesis applied at $u$ gives
\begin{align*}
    y_vV_v - h_v\dot y_v&= y_uV_v +(y_v-y_u)V_v - h_v\dot y_v\\
    &=y_uV_v +h_v(\dot y_v - \dot y_{u}) - h_v\dot y_v \tag{Ohm's law on $(c,v)$ for $c\in C_v\setminus\{u\}$}\\
    &= y_uV_v -h_v \dot y_{u}\\
    &=y_uV_u -h_u \dot y_{u}  + y_u(V_v-V_u)- \rho_u \dot y_u\tag{$h_v = h_u+\rho_u$}\\
    &= y_uV_u -h_u \dot y_u\tag{Ohm's law on $(u,v)$}\\
    &=y_lV_l \tag{induction hypothesis}.
\end{align*}
\end{proof}

\begin{lemma}[Bounding movement by $y_l$]\label{lem:open-circuit} In an elongation of edge $(p_l,l)$ at unit rate, the instantaneous movement cost in the approximate edge length metric satisfies
\[
\begin{aligned}
\|\dot{\by}\|_{\rho}
\le\frac{2y_l^p}{\varepsilon}
\log\!\left(\frac{k}{y_l}\right).
\end{aligned}
\tag{E}\label{eq:elongation}
\]
    
\end{lemma}

\begin{proof}
Consider a strict ancestor $v\succ l$ and a child $c$ such that $l \not \in \Lcal_c$, the current flowing from $v$ to $\Lcal_c$ is 
\begin{equation}
    \dot y_c = V_v \frac{y_c}{h_v}\label{eq:cur-c}
\end{equation}
by the effective resistance argument. Therefore, denoting by $u(v)$ the child of $v$ that is an ancestor of $l$, we have, 
\begin{align*}
    \|\dot{\by}\|_{\rho}
:&=\sum_{v\ne r}\rho_v|\dot y_v|\\
&= 2\sum_{v\succ l} h_{v}\sum_{c\in C_v\setminus\{u(v)\}}\dot y_{c}\tag{$\rho$ defines an ultrametric with heights $h_v$}\\
&=2\sum_{v\succ l}\bigl(y_v-y_{u(v)}\bigr)V_v.\tag{using \eqref{eq:cur-c}}
\end{align*}
Using the bound from Lemma \ref{lem:voltage}, $V_v \leq \frac{y_l}{y_v}V_l$, we have,
\begin{align*}
\|\dot{\by}\|_{\rho}
&\le\frac{2y_l^p}{\varepsilon}
\sum_{v\succ l}\left(1-\frac{y_{u(v)}}{y_v}\right)\nonumber\\
&\le\frac{2y_l^p}{\varepsilon}
\log\!\left(\frac{k}{y_l}\right).
\end{align*}
which is precisely \eqref{eq:elongation}. The last step uses $1-a\le-\log a$, and a telescopic sum.
\end{proof}

Lemma~\ref{lem:geometry} gives $d_u\le2\rho_u$, and $y_l\ge1/2$ by \eqref{eq:invariants}. Since $\dot \Phi=y_l^p$, we can conclude from Lemma \ref{lem:open-circuit} that,
\[
\dot C_{\by}:= y_l + \|\dot{\by}\|_{d}
\le y_l +  \frac{4y_l^p}{\varepsilon}
\log\!\left(\frac{k}{y_l}\right) \le\left(2 +\frac{4\log(2k)}{\varepsilon}\right) \dot \Phi\le \frac{32\log(2k)}{\varepsilon}\dot \Phi,
\tag{Elongation Eq.}\label{eq:physical-elongation}
\]
where we used $y_l \leq 2 y_l^p$ because $y_l\geq 1/2$.

The following lemma will be used in the next section.

\begin{lemma}[Bounding movement by $\dot y_l$]\label{lem:rate2} In an elongation of leaf $l$, if $y_{p_l}\geq 2y_l$, then, 
    \[
    ||\dot \by||_\rho \leq 4\rho_l(-\dot y_l)\log\left(\frac{k}{y_l}\right) \tag{E'}\label{eq:ep}
    \]
\end{lemma}
\begin{proof}
Since $y_{p_l}\geq 2y_l$, applying Lemma \ref{lem:voltage} at $p_l$ gives $V_{p_l}\leq V_l/2$ and thus, applying Ohm's law in edge $(l,p_l)$, as well as using $V_l = y_l^\varepsilon/\epsilon$ gives,
    \[
    -\dot y_l = \frac{y_l}{\rho_l}(V_l - V_{p_l})\geq \frac{y_l}{2\rho_l}V_l \geq \frac{ y_l^p}{2\varepsilon \rho_l}.
    \]
Finally, invoking
\[
\begin{aligned}
\|\dot{\by}\|_{\rho}
\le\frac{2y_l^p}{\varepsilon}
\log\!\left(\frac{k}{y_l}\right).
\end{aligned}
\tag{E}
\]
we recover \eqref{eq:ep}.
\end{proof}

\subsection{Leaf deletions}\label{sec:deletion}
In this section, we show that \eqref{eq:goal} is preserved during the deletion of a leaf $l$, i.e., that
\[
\boxed{\Delta C_y := 
\OT_{d}\bigl(\by(t^-),\by(t^+)\bigr)
\le\frac{32\log(2k)}{\varepsilon}\Delta \Phi.
}
\tag{Deletion Eq.}\label{eq:deletion}
\]
where $t^-$ and $t^+$ denote the instants immediately before and after the operation and $\Delta \Phi=\Phi(t^+)-\Phi(t^-).$

Without loss of generality, we assume $y_l^-=1/2$ before the deletion. Indeed, when a leaf $l$ is deleted with $y_l^->1/2$, one can first artificially extend the leaf until $y_l^-=1/2$ while controlling the movement cost with \eqref{eq:elong}, and then delete $l$. 

\begin{lemma}[Lower bound on the energy increase]\label{lem:lben}
In the deletion of a leaf $l$, we have,
\begin{equation}
    \Delta \Phi\ge\varepsilon d_l y_l(t^-)^p.
\label{eq:bregman}
\end{equation}
\end{lemma}

\begin{proof}
$\Delta\Phi=\Phi(\by^+)-\Phi(\by^-)$.
By the first-order condition, since both configurations have total mass $k$, we have $\langle\nabla\Phi(\by^-),\by^+-\by^-\rangle=0$.
Subtracting this zero linear term gives
\[
\begin{aligned}
\Delta\Phi
&=\sum_{u\ne r}d_u\bigl[(y_u^+)^p-(y_u^-)^p-p(y_u^-)^{p-1}(y_u^+-y_u^-)\bigr]\\
&\ge d_l\bigl[-(y_l^-)^p+p(y_l^-)^p\bigr]
=\varepsilon d_l(y_l^-)^p,
\end{aligned}
\]
where we used the convexity of $z\mapsto z^p$, to imply that every summand is nonnegative,$\forall y_1,y_2\geq 0 :
y_1^p - y_2^p -py_2^{p-1}(y_1-y_2)\geq 0$, and we retained only the deleted edge (for which $y_l^+ = 0$).
\end{proof}

\begin{lemma}[Upper bound on the movement]\label{lem:ubmov} Assuming that $y_l = 1/2$ upon deletion
\begin{equation}
\OT_d\bigl(\by(t^-),\by(t^+)\bigr)
\le 4d_l\bigl(\log(2k)+1\bigr).\label{eq:deletion-integral}
\end{equation}
\end{lemma}

\begin{proof}  
Consider the following optimization problem,
\[
\by(s)=\arg\min_{\substack{\by\ge0\\\sum_{\ell}y_{\ell}=k}}
\bigl\{\Phi_T(\by)+s y_l^p\bigr\},\qquad s\in [0,\infty),
\]
The variable $s\geq 0$ is an interpolation parameter that extends the edge $(p_l,l)$ by $s$ and defines an auxiliary tree $T(s)$ with edge lengths
\[
d_l(s)=d_l(t^-)+s,
\qquad  d_u(s)=d_u(t^-)\quad(\forall u\ne l)
\] 
where we overload the parameter of $d_u(\cdot)$ with $s$. When unparametrized, the notation $d_l$ refers to the original edge length $d_l(t^-)$. The configuration $\by(s)$ is the ordinary minimizer of \eqref{eq:optim} on the tree $T(s)$ and satisfies $\by(0)=\by(t^-)$, and $\lim_{s\rightarrow \infty}\by(s)=\by(t^+)$ because $\Phi(\by(s)) + sy_l(s)^p\leq \Phi(\by(t^+))$ implies that $y_l(s) \rightarrow 0$. Defining $\rho_u(s) =  d_u(s) y_u^\varepsilon(s)$ and observing that $y_{p_l}(s)\geq 2y_l(s)$, we apply Lemma \ref{lem:rate2} which implies, for $s\geq 0$,
\[
    ||\dot \by||_{\rho} \leq 4 \rho_l(s)(-\dot y_l(s))\log\left(\frac{k}{y_l(s)}\right) \tag{E'}.
\]
Furthermore, we will show that throughout the elongation, the following identities remain valid
\begin{equation}
     \forall u: d_u \leq 2 \rho_u(s),  \quad \text{and}\quad \rho_l(s) \le d_l, \label{eq:near-metric}
\end{equation}
despite \eqref{eq:invariants} being violated at $l$. Consequently,
\[
||\dot \by||_{d} \leq 8d_l(-\dot y_l)\log\left(\frac{k}{y_l}\right).
\]
Integrating this quantity between $s= 0$ and $s=\infty$, with the change of variable $-\int_{s=0}^{\infty}\dot y_l\log\left(\frac{k}{y_l}\right)ds = \int_{y_l=0}^{1/2}\log\left(\frac{k}{y_l}\right)dy_l = \frac{1}{2}(\log(2k) + 1)$ gives the desired upper bound
\[\OT_d\bigl(\by(t^-),\by(t^+)\bigr) = \int_{s=0}^{\infty}||\dot \by||_dds\leq 4d_l(\log(2k) + 1).\]

We conclude the proof by justifying \eqref{eq:near-metric}. For every $u\neq l$, throughout the deletion $y_u \geq 1/2$, so we have \[
\forall u\neq l : \rho_u(s) = d_u y_u^\epsilon(s) \geq d_u/2.
\]
It remains to argue about $\rho_l$. We choose a surviving leaf $\ell$ below a sibling of $l$. For every $u\in\ell\rightarrow p_l$, we have $y_u(0)\le y_u(s)\le2y_u(0)$, because the initial mass is at least $1/2$, and a mass of at most $1/2$ is displaced. Applying \eqref{eq:FO} in the auxiliary tree gives for all $s\geq 0$,
\[
 \rho_l(s) = \sum_{u\in\ell\rightarrow p_l}\rho_u(s),
\]
and the right-hand side varies within a factor of $2^\epsilon$, and is equal to $\rho_l(s=0) = \frac{1}{2^\epsilon}d_l$ at initialization, thus,
\[
\frac12d_l
\le\rho_l
\le d_l.
\]
\end{proof}

Using \eqref{eq:bregman} for $y_l = 1/2$ gives $\Delta  \Phi\ge\varepsilon d_l(1/2)^p\ge\varepsilon d_l/4$. Therefore, the right-hand side of \eqref{eq:deletion-integral} is at most
\[
\frac{16}{\varepsilon}\bigl(\log(2k)+1\bigr)\Delta \Phi
\le\frac{32\log(2k)}{\varepsilon}\Delta \Phi,
\]
which proves \eqref{eq:deletion}.

\subsection{Leaf forks and finishing the proof of Proposition \ref{lem:potential}}\label{sec:forks}

Finally, we consider a fork of leaf $l$ at time $t^-$, creating $m$ children with common
edge length $\delta$ at time $t^+$. We interpolate this operation by first
splitting the fractional mass equally among the new children
at edge length zero, then increasing their common edge length
$s$ from $0$ to $\delta$. By symmetry, the new children have equal fractional masses. The initial split (of length zero) has zero cost and preserves $\Phi$. By the choice of $\delta$, every leaf has mass at least $1/2$
throughout the interpolation; therefore,  integrating \eqref{eq:elong} for $s\in(0,\delta]$ we obtain,

\[
\boxed{\Delta C_y := 
\OT_{d}\bigl(\by(t^-),\by(t^+)\bigr)
\leq \frac{32\log(2k)}{\varepsilon}\Delta \Phi.
}
\tag{Fork Eq.}\label{eq:fork}
\]

\begin{proof}[Proof of Proposition \ref{lem:potential}]
The three equations \eqref{eq:physical-elongation}, \eqref{eq:deletion}, and \eqref{eq:fork} show that \eqref{eq:goal} is preserved throughout a run of the continuous tree-mining game, therefore finishing the proof of Proposition~\ref{lem:potential}.
\end{proof}


\section*{Conclusion}\label{sec:conclusion}
We showed that asynchronous collective tree exploration admits a regret guarantee of the form
\[
\frac{2n}{k}+O\!\left(\log^2(2k)D\right),
\]
matching the $\Omega(\log^2 k)$ lower bound of \cite{cosson2025asynchronous}.
The proof uses a multiscale power regularizer to define a continuous configuration, which is rounded into a tree-mining strategy. Because of the tree-traversal reduction described in \cite{cosson2024collective}, the approach incidentally also yields a new $O(w^2)$-competitive algorithm for width-$w$ layered graph traversal, matching the tight guarantee of \cite{bubeck2022shortest,bubeck2023randomized}, albeit with a power regularizer.

\paragraph{AI disclosure.}
Work on this project began in November 2023, immediately after we
obtained a $O(k)$ regret guarantee using a quadratic
regularizer ($\epsilon=1$)~\cite{cosson2024collective}.
In 2024, we established the $\Omega(\log^2 k)$ lower bound
for asynchronous collective tree exploration subsequently which appears
in~\cite{cosson2025asynchronous}. We then obtained an
$O(\log^3 k)$ competitive guarantee for synchronous collective tree
exploration using an entropic regularizer (corresponding to the
limit $\epsilon\to0$). We developed this result into a complete
manuscript that remained unpublished; an intermediate version of which appeared
in Chapter~8 of RC's thesis, defended in July
2025~\cite{cosson2025collective}. All of these results, including
the first polylogarithmic competitive ratio for collective tree
exploration, were obtained without AI assistance. In September~2026, discussions with OpenAI's Astra model led to
the $O(\log^2 k)$ regret bound presented in this paper,
whose dependence on $k$ is optimal in the asynchronous setting. 
The approach builds on our earlier framework and follows the
overall structure of~\cite{cosson2024collective}. A key
contribution of the model was to propose the intermediate scaling
$\epsilon=\Theta(1/\log k)$, rather than the entropic limit,
which led to unexpectedly simpler computations.
After verifying the correctness of the approach, we revised the proof substantially to simplify and clarify
its exposition.

\newpage
\appendix
\section{Proof of Proposition \ref{lem:rounding}}
\label{sec:rounding}

\roundinglemma*

\begin{proof}
The repair rule is well-defined because the root satisfies $x_r = y_r$, so the stopping ancestor \(v\) exists. Its child $u$ towards $\ell$ satisfies $x_u\geq y_u+3/2$ whereas $x_v< y_v+3/2$, there must exist a child $c$ of $v$ with $x_c<y_c$, and the path to the leaf $\ell'$ is well-defined. 

The algorithm preserves \eqref{eq:invariants}. During elongation of leaf $l$, only $y_l$ decreases, and by the time it reaches $1/2$, a repair
enforces $x_l=1$, blocking further elongation. Deletions increase every surviving fractional leaf mass. In forks, $\delta>0$ is chosen sufficiently small so as to maintain \eqref{eq:invariants}.
This is possible because, denoting by $m$ the number of children of $\ell$,
\[
\frac{y_\ell}{m}>\frac34,
\qquad
\left\lceil\frac{x_\ell}{m}\right\rceil-\frac{y_\ell}{m}<\frac54,
\]
using $m\le x_\ell-1$ and $x_\ell-y_\ell<3/2$.
These strict margins persist for sufficiently small $\delta$
by continuity of the fractional minimizer.

The proof of \eqref{eq:ybx} relies on a potential function argument. For a tree $T$ and configurations $\bx,\by$, define
\[
P(T,\bx,\by)
=8\sum_{u\ne r}d_u\left(x_u-y_u-\frac12\right)_+,
\qquad
P(t)=P(T(t),\bx(t),\by(t)).
\]
This potential is nonnegative and vanishes when $\bx=\by$.
Since the positive-part function is $1$-Lipschitz, changing either
configuration changes $P$ by at most eight times the movement cost. We will show that the following inequality, which implies \eqref{eq:ybx},
\[
C_{\bx}(t)+P(t)\le128C_{\by}(t), \tag{Pot}
\]
remains true during a) repair b) elongation c) fork d) deletion. 
Summing their contributions and using
$C_{\bx}(0)=C_{\by}(0)=P(0)=0$ proves the lemma.

\textit{a) Repair.} Every edge on the upward path $\ell\rightarrow v$ is overloaded by at least $3/2$,
so moving one miner up along edge $(p_u,u)$ decreases $P$ by exactly
$8d_u$. Every edge on the downward path $\ell' \rightarrow v$ is underloaded, so adding
one miner increases its contribution by at most $4d_u$.
Using \eqref{eq:geometry},
\[
\begin{aligned}
\Delta C_x(t)+\Delta P
&\le d(\ell,\ell')-8d(\ell,v)+4d(v,\ell')\\
&=-7d(\ell,v)+5d(v,\ell')\\
&\le-\frac13d(\ell,v)<0 = \Delta C_y(t).
\end{aligned}
\]

\textit{b) Elongation.} During an elongation $\dot C_y = y_\ell + \|\dot{\by}\|_d$, and the potential increase satisfies,
\[
\dot P
\le8\left(x_\ell-y_\ell-\frac12\right)_+
+8\|\dot{\by}\|_d.
\]
Since $x_\ell\le4y_\ell$, $y_\ell\ge1/2$, and
$x_\ell-y_\ell\le3/2$,
\[
\begin{aligned}
\dot C_{\bx}+\dot P
&\le x_\ell
 +8\left(x_\ell-y_\ell-\frac12\right)_+
 +8\|\dot{\by}\|_d\\
&\le20y_\ell+8\|\dot{\by}\|_d
\le20\dot C_{\by}.
\end{aligned}
\]

\textit{c) Fork.} During a fork, the split changes neither
potential nor cost, then the elongation to length $\delta$ is covered by the elongation analysis.

\textit{d) Deletion.} During a deletion of $\ell$, the deleted edge alone leads to
\[
\Delta C_{\by}
=\OT_d\bigl(\by(t^-),\by(t^+)\bigr)
\ge d_\ell y_\ell(t^-).
\]
By \eqref{eq:geometry}, the cost of evacuating the miners to the chosen leaf
below a sibling is bounded by
\[
3d_\ell x_\ell(t^-)
\le12d_\ell y_\ell(t^-)
\le12\Delta C_{\by}.
\]
Evacuation also increases $P$ by at most $8$ times this quantity, and updating the
fractional configuration increases $P$ by at most
$8\Delta C_{\by}$, hence
\[
\Delta C_{\bx}+\Delta P
\le9\times 12 \Delta C_{\by} +8\Delta C_{\by}
\le116\Delta C_{\by}.
\]
Subsequent repairs have nonpositive amortized cost.
\end{proof}

\bibliographystyle{alpha}
\bibliography{biblio}

\end{document}

%% file: packages.tex
\usepackage[utf8]{inputenc} 
\usepackage{url}            
\usepackage{booktabs}       
\usepackage{amsfonts}       
\usepackage{nicefrac}       
\usepackage{microtype}      

\usepackage{amsthm}
\usepackage{amsmath}
\usepackage{amssymb}
\usepackage{comment}        
\usepackage{enumitem}  
\usepackage{cleveref} 

\usepackage{thmtools}

\newtheorem{theorem}{Theorem}[section]

\newtheorem*{theorem*}{Theorem}
\newtheorem{proposition}[theorem]{Proposition}

\newtheorem*{remark*}{Remark}
\newtheorem{lemma}[theorem]{Lemma}